\documentclass[11pt]{article}

\usepackage{arxiv}
\usepackage{amsmath}
\usepackage{amsthm}
\usepackage[utf8]{inputenc} 
\usepackage[T1]{fontenc}    
\usepackage{hyperref}       
\usepackage{url}            
\usepackage{booktabs}       
\usepackage{amsfonts}       
\usepackage{nicefrac}       
\usepackage{microtype}      
\usepackage{cleveref}       
\usepackage{lipsum}         
\usepackage{graphicx}
\usepackage{doi}
\usepackage{enumitem} 

\usepackage{adjustbox}
\usepackage[ruled, vlined, linesnumbered, noresetcount]{algorithm2e}
\newtheorem{definition}{Definition}
\newtheorem{lemma}{Lemma}
\newtheorem{theorem}{Theorem}
\newtheorem{corollary}{Corollary}

\newcommand{\remove}[1]{}

\newcommand{\N}{{\mathbb{N}}}

\def \N {\mathbb{N}}
\def \demand {demand}
\def \radius {radius around}
\def \hH {{{\hat{H}}}}
\def \hG {{{\hat{G}}}}
\def \G {{{G}}}
\def \H {{{H}}}
\def \true {{\tt true}}
\def \false {{\tt false}}
\def \n {{p}}
\def \hn {{{\hat{p}}}}

\def \hS {{{\hat{S}}}}
\def \hb {{{\hat{b}}}}
\def \d {d}
\def \t {t}
\def \htt {{{\hat{t}}}}
\def \r {r}
\def \s {s}
\def \q {q}
\def \bd {{\bf d}}
\def \one {{\bf 1}}
\def \bt {{\bf t}}
\def \bu {{\bf 1}}
\def \hbt {{{\hat{{\bf t}}}}}

\def\tw{{\tt{tw}}}
\def\nd{{\tt{nd}}}
\def\cw{{\tt{cw}}}
\def\mw{{\tt{mw}}}
\def\bw{{\tt{bw}}}

\def\dmax{\delta}
\def\tmax{\tau}

\def \DVD {{\sc DVD}}
\def \VD {{\sc VD}}
\def \RD{{\sc RD}}
\def\DVDl{Distance Vector Domination}
\def\VDl{Vector Domination}
\def\RDl{R-Domination}

\def \Ecd {{E_{cd}}}
\def \Vc {{V_c}}

\def \Lc {{L_c}}
\def \Ld {{L_d}}
\def \Lcn {{\Lc\mbox{-neg}}}
\def \Lcp {{\Lc\mbox{-pos}}}
\def \Lcg {{\Lc\mbox{-guard}}}
\def \ecdo {{e_0^{cd}}}

\def \ecdm {{e_\s^{cd}}}
\def \ecdj {{e_j^{cd}}}
\def \vco {{v_0^c}}
\def \vcu {{v_1^c}}
\def \vcr{{v_r^c}}
\def \vci {{v_i^c}}
\def \vch {{v_h^c}}
\def \Lcd {{L_{cd}}}
\def \Lcdn {{\Lcd\mbox{-neg}}}
\def \Lcdp {{\Lcd\mbox{-pos}}}
\def \Lcdg {{\Lcd\mbox{-guard}}}
\def \Mcd {{M_{cd}}}
\def \Mcdn {{\Mcd\mbox{-neg}}}
\def \Mcdp {{\Mcd\mbox{-pos}}}
\def \Mcdg {{\Mcd\mbox{-guard}}}
\def \Iccd {{I_{c:cd}}}
\def \Idcd {{I_{d:cd}}}
\def \Iccdp {{\Iccd\mbox{-pos}}}
\def \Iccdn {{\Iccd\mbox{-neg}}}
\def \Iccdg {{\Iccd\mbox{-guard}}}
\def \Idcdp {{\Idcd\mbox{-pos}}}
\def \Idcdn {{\Idcd\mbox{-neg}}}
\def \Idcdg {{\Idcd\mbox{-guard}}}
\def \MQ {{\sc MQ}}
\def \bug{bag}

\newcommand{\Ni}[2]{N_{#1}(#2)}

\newtheorem{observation}[theorem]{\textbf{Observation}}

\title{Distance Vector Domination}

\date{September 9, 2024}
\date{}

\newif\ifuniqueAffiliation

\ifuniqueAffiliation 
\author{ 
\href{https://orcid.org/0000-0000-0000-0000}{
\hspace{1mm}Gennaro Cordasco
\\
Department of Psychology, \\
University of Campania ''L.Vanvitelli'', Italy. \\ \texttt{gennaro.cordasco@unicampania.it}\\
	\And
	\href{https://orcid.org/0000-0000-0000-0000}
    Luisa Gargano, Adele A. Rescigno \\
	Department of Computer Science, \\
    University of Salerno, Italy.\\
	\texttt{\{lgargano,arescigno\}@unisa.it} \\

}
\else
\usepackage{authblk}

\setlength{\affilsep}{0em}
\newbox{\orcid}\sbox{\orcid}{\includegraphics[scale=0.06]{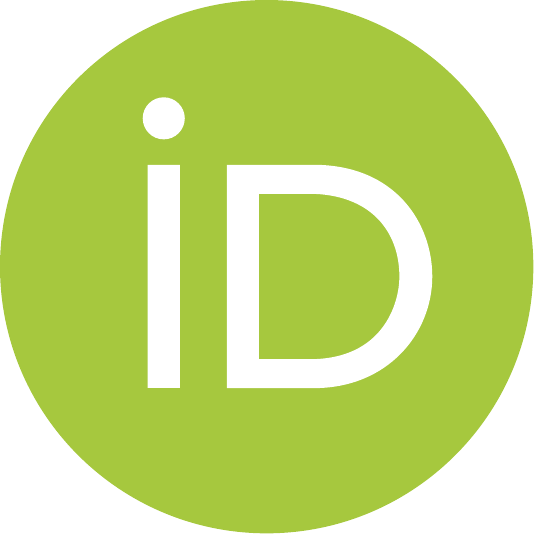}} 
\author[1]{%
	\href{https://orcid.org/0000-0001-9148-9769}
   {\usebox{\orcid}\hspace{1mm}Gennaro Cordasco
   }%
}
\author[2]{%
	\href{https://orcid.org/0000-0003-3459-1075}{\usebox{\orcid}\hspace{1mm}Luisa Gargano\thanks{\texttt{
Work  partially supported by project SERICS (PE00000014)
under the MUR National Recovery and Resilience Plan funded by the
European Union - NextGenerationEU.}}}%
}
\author[2]{%
	\href{https://orcid.org/0000-0001-9124-610X}{\usebox{\orcid}\hspace{1mm}Adele A. Rescigno$^*$}%
}
\affil[1]{Department of Psychology, 
University of Campania ''L.Vanvitelli'', Italy, {\tt{gennaro.cordasco@unicampania.it}}}
\affil[2]{Department of Computer Science, 
    University of Salerno, Italy, {\tt{lgargano@unisa.it}}}
    \affil[3]{Department of Computer Science, 
    University of Salerno, Italy, {\tt{arescigno@unisa.it}}}

\renewcommand{\shorttitle}{Distance Vector Domination}

\hypersetup{
pdftitle={{Distance Vector Domination}},
pdfsubject={},
pdfauthor={Gennaro Cordasco, Luisa Gargano and Adele A. Rescigno},
pdfkeywords={Distance Domination, Vector Domination, Parameterized complexity,
  Neighborhood diversity, Modular-width, Treewidth},
}

\begin{document}
\maketitle

\baselineskip=0.5truecm
\begin{abstract}
 Identifying and mitigating the spread of fake information is a challenging issue that has become dominant with the rise of social media. We consider a generalization of the Domination problem that can be used to detect a set of individuals who, once immunized, can prevent the spreading of fake narratives. The considered problem, named {\em Distance Vector Domination} generalizes both distance and multiple domination, at individual (i.e., vertex) level. We study the parameterized complexity of the problem according to several standard and structural parameters.  We prove the W[1]-hardness of the problem with respect to neighborhood diversity, even when all the distances are  $1$. We also give fixed-parameter algorithms for some variants of the problem and parameter combinations. 
\end{abstract}


\baselineskip=0.7truecm
\section{Introduction} 
\label{intro}
 Domination is a fundamental concept in graph theory, which deals with the idea of dominating sets within a graph. In this problem, you seek to find the smallest set of vertices in a graph in such a way that every vertex in the graph is either in the dominating set or adjacent to a vertex in the dominating set. Dominating sets are critical in various real-world applications across fields that involve networks, connections, and coverage \cite{BJP,sensor1,sensor2,CGL+,sensor3}.
  \\
 In social network analysis, it can be   used to identify key influencers or individuals who can drive information diffusion, trends, or behaviors within a network \cite{CGL+
 }. It helps businesses and researchers target their efforts efficiently \cite{BJP}.
 \\
 In wireless sensor networks, where sensors are placed to monitor an environment or collect data \cite{sensor3,sensor2,sensor1}, the Domination Problem helps minimize the number of sensors while ensuring full coverage and connectivity.

 In facility location and placement, it helps determine the optimal locations for services, facilities, or resources to ensure that they are accessible to a maximum number of people while minimizing the number of locations needed.

The Domination Problem in graph theory has several important variants that focus on different aspects of the problem. We focus on generalizations to distance domination and multiple domination.
\textsl{Multiple Domination}: Every vertex in the graph must  either be part of a dominating set or adjacent to a prescribed number of vertices in a dominating set. The goal is to minimize the size of the dominating set. 
 \textsl{Distance Domination}:  For any vertex $v$ not in the dominating set, there must be at least one vertex within a specified distance  from $v$ that is part of the dominating set; the goal is to minimize the size of the dominating set.  
Distance and multiple domination provide ways to address practical concerns related to the physical or operational limitations of networks.  Application areas can include wireless sensor networks, facility location, communication networks, and transportation network design.
We refer to \cite{HHS1,HHS2} for a survey of domination in graphs and its several variations.

This paper is motivated by an application aimed at combating the spread of misinformation using epidemiological principles. Graph-based information diffusion algorithms offer a means to analyze the dissemination of both genuine and fake information within a network. Sociologists widely employ threshold models to characterize collective behaviors \cite{granovetter}, and their application to studying the propagation of innovations through networks was initially proposed in \cite{Kempe}.
The linear threshold model has then been extensively employed  in the literature to study influence maximization,  a critical   problem    in network analysis, which aims at identifying a small subset of vertices capable of maximizing the diffusion of content throughout the network 
\cite{Ben-Zwi,CCG+,SNAM,CGMRVa,CGMRV,asonam,CGRV,active,itp}.
While these algorithms are not designed for intercepting fakes, they can be used as a component in a broader strategy for identifying and mitigating the spread of fake information.

Controlling the spread of misinformation/disinformation is an ongoing challenge. Strategies for reducing the spread size by either blocking some links, so that they cannot contribute to the diffusion process \cite{Kimuraetal}, or by immunizing/removing vertices were considered in several papers \cite{barabasi,newman}.
In this paper, we focus on the second strategy: limit the spread by immunizing a bounded number of vertices in the network.
We consider a population of interconnected individuals that can potentially be misinformed by a word-of-mouth diffusion strategy. We assume that when an individual is reached by a sufficient amount of debunking information, he/she becomes immunized. With more details, an individual gets immunized if he/she receives the debunking information from a number of neighbors at least equal to its threshold. Moreover, each individual has a certain level of trust in the others (circle of trust) described by a radius around it. Only debunking information coming from within the circle of trust is considered reliable.
In particular, we consider the use of generalized dominating sets to detect a group of individuals who, by setting a debunking (or prebunking) campaign, can prevent the spreading of negative narratives. 
In the presence of a debunking (or prebunking) campaign,  the \emph{immunization} operation on a vertex inhibits the contamination of the vertex itself. Thus, to avoid the diffusion of malicious items, we are looking for a small subset $T$ of vertices (immunizing set) that, by spreading the debunking information, enables us to  {stop} the misinformation diffusion. The immunizing set should be able to cover each vertex multiple times (based on the vertex threshold)  within a maximum distance (depending on the radius that specifies the circle of trust of the vertex). 
We propose the \DVDl\ problem,  which includes both multiplicity and distance. 
\def\dist{\delta}

\smallskip
\noindent \textbf{The Problem.} 
Let $G=(V,E)$ be a graph. We denote by $n=|V|$ the number of vertices in $G.$ For a set of vertices  $X\subseteq V$,  we denote by $G[X]$  the induced subgraph of $G$ generated by $X$.
Given two vertices  $u,v\in V$, we denote by $\dist_G(u,v)$ the distance between  $u$ and $v$ in $G$. 
Moreover, for a vertex $v\in V$, we denote by $\Ni{G}{v}=\{u\in V \mid  (u,v)\in E\}$ the 
neighborhood of $v$  and  by $\Ni{G,\d}{v}=\{u \in V \mid u \neq v \land \dist_G(u, v)\leq  \d\}$ the {\em neighborhood of radius $\d$ around $v$}. Clearly, $\Ni{G,1}{v} = \Ni{G}{v}$.
We also define the distance between a vertex $v\in V$ and a set $U\subseteq V$ as $\dist_G(v,U)=\min_{u\in U}\dist_G(v,u).$
We  
omit the subscript $G$ whenever the graph is clear from the context.
\begin{definition}
Given a 
graph $G=(V,E)$ and   vectors $\bt=(\t_v \mid \t_v\in \N, \, v \in V)$ and 
$\bd=(\d_v \mid \d_v\in \N,\,v \in V )$, 
where $\t_v\leq |\Ni{\d_v}{v}|$, a  Distance Vector Dominating set  $S$  is a set $S \subseteq V$ such that
 $|\Ni{\d_v}{v}\cap S|\geq \t_{v}$,   for all $v \in V\setminus S$.
\end{definition}
We will consider the following problem.
\begin{quote}
\textsc{\DVDl\ (\DVD)}:\\
\textbf{Input:} A 
graph $G=(V,E)$,   vectors $\bt=(\t_v \mid \t_v\in \N,\, v\in V)$  and \\
 \hphantom{Input: }    $\bd=(\d_v \mid \d_v\in \N,\, v \in V)$.\\ 
\textbf{Output:} A Distance Vector Dominating set of minimum size.
\end{quote}
For each vertex $v \in V$,  we will refer to $\t_v$ and $\d_v$ as the {\em \demand} of $v$ and the {\em \radius} $v$, respectively.
Furthermore, given a set $S \subseteq V$, we say that a vertex $v \in V \setminus S$ is {\em dominated by $S$} if $|\Ni{\d_v}{v}\cap S|\geq \t_{v}$.
{
Since the problem can be solved independently in each connected component of the input graph, from now on, we assume that the input graphs are connected. }

 \DVD\  generalizes several well-known and widely studied problems. Consider an input graph $G=(V,E)$ together with   vectors  $\bt=(\t_v \mid \t_v\in \N,\,v \in V)$ and $\bd=(\d_v \mid \d_v\in \N,\,v \in V)$: 
\begin{itemize} 
\item 
When
$\bt=\bd=\one=(1,\ldots,1),$ \DVD\   becomes the classical \textsc{Dominating Set (DS)} problem \cite{HHS1}. 
\item
If  $\bd=\one$
and $\bt=(\t_v \mid \t_v\in \N,\, v\in V)$, 
then \DVD\  becomes the \textsc{\VDl\ (\VD) } problem,
which asks for a minimum size   set $S \subseteq V$  such that $|N(v)\cap S|\geq \t_{v}$,   for each $v \in V\setminus S$. \VDl, introduced in \cite{GH},  has been extensively studied \cite{CMV,HPV,MP,LWS} and
 was recently studied from the parameterized complexity point of view in \cite{IOU,RSS}.
The special case  $\bt=(r,\ldots, r)$ for some positive integer $r$ has been studied under the name of 
\textsc{$r$-Domination} \cite{Flink85}.
\item
The problem corresponding  to the  case   
$\bt=\one$ and   $\bd=(\d_v \mid \d_v\in \N,\, v\in V)$ was introduced  by Slater in \cite{Slater} under the name of
\textsc{\RDl\ (\RD)}.
The special case   
$\bd=(d,\ldots, d)$, for some positive integer $d$, has been studied under the name of  \textsc{Distance Domination (DD)}
 \cite{Henning20}. 
\end{itemize}
 Knowing that \DVD\ generalizes the \VD\ problem \cite{CMV}, we immediately have  
 \begin{theorem}
 \DVD \  cannot be approximated in polynomial time  to within a factor of  {$0.2267\log n$}
 ,  unless P=NP.
 \end{theorem} 
Moreover, following the lines of the proof of Theorem 1 in \cite{CMV}, one can easily get a logarithmic approximation algorithm. 
 %
\begin{theorem} \label{th_approx}
\DVD\ can be approximated in polynomial time by a factor 
$\log n+2.$
\end{theorem}
%

\def\XP{{\em XP}}
\subsection{Parameterized Algorithms}
Parameterized complexity is a refinement to classical complexity  in which one takes into account, not only the input size but also other aspects of the problem given by a parameter $p$. 
We recall that a problem with input size $m$ and parameter $p$ is called {\em fixed parameter tractable (FPT)} if it can be solved in time $f(p) \cdot m^c$, where $f$ is a computable function only depending on $p$ and $c$ is a constant.  {A problem is in \XP\ parameterized by $p$ if it can be solved in time $m^{f(p)}$, where $f$ is a computable function only depending on $p$.}

\smallskip
\noindent {\bf Known  results.} 
The \DVD\ problem, as well as each of the special cases described in Section \ref{intro},  is  W[2]-hard with respect to the size $k$ of the solution since they all generalise the W[2]-complete {\sc Dominating Set} problem \cite{DF}.

It is shown in \cite{Betzler2012} that  \VD\  is W[1]-hard with respect to treewidth, thus implying that \DVD\ is W[1]-hard with respect to treewidth.
{Recently, Lafond and Luo \cite{LL} presented an FPT algorithm for  
\textsc{$r$-Domination}  parameterized by neighborhood diversity and proved that this problem is W[1]-hard with respect to modular-width.}
In \cite{Borradaile2016} the authors consider the \textsc{DD} problem parameterized by the radius ($d$) and the treewidth ($\tw$). They show an FPT ($O((2d+1)^\tw n)$) algorithm. Moreover, they also show that the running time dependence on $d$ and $\tw$ is the best possible under Strong Exponential Time Hypothesis (SETH) \cite{SETH}.
This lower bound applies also to the \RD\ problem which generalizes  \textsc{DD}.
{In \cite{KLP19} the authors study the ($k,r$)-center problem, which, given a graph $G $, asks if there exists a set $K$ of at most $k$ vertices of $G$, so that $\min_{v \in K} \delta(v, u) \leq r$ for each $ u \notin  K$. The ($k,r$)-Center problem represents the decision version of the \textsc{DD} problem and is W[1]-hard when parameterized by {\tt{fvs}} + $k,$ where {\tt{fvs}} represents the feedback vertex set parameter \cite{KLP19}. Since the treewidth of a graph $G$ is upper bounded by the feedback vertex set number of $G$ plus one, this result implies that    \RD\   is W[1]-hard when parameterized by the treewidth.}

A \XP\ algorithm, with running time $O(n+1)^{O(\cw)}$,  for a generalization of   \textsc{VD} on graphs of bounded clique-width (\cw) has been provided in \cite{CCG+}. Since $\cw \leq 2^ {\tw+1} + 1$ \cite{CO00}, this result implies the \XP\ solvability of the \textsc{VD}  problem for graphs of bounded treewidth.
{Assuming to have a branch decomposition of the input graph of width \bw, a FPT algorithm with running time $O((\tmax+2)^\bw [(\tmax +1)^2+1]^{\bw/2} \, n^2)$ for \VD\ is given  in \cite{IOU}, where $\tmax$ is the largest demand of the vertices, i.e.,  $\tmax=\max_{v \in V} \t_v$. Since $\max\{\bw,2\} \leq \tw +1 \leq \max\{3 \bw/2,2\},$ this result  implies an  $O((\tmax+2)^{\tw{+1}} [(\tmax +1)^2+1]^{(\tw{+1})/2} \, n^2)$ time algorithm for \VD.} 
\\
{An  algorithm for {\sc Dominating Set} problem parameterized by modular-width (\mw) was given by Romanek \cite{R}; it requires $O(2^\mw \, n^2)$ time. However, little work has been done to design FPT algorithms for  \textsc{DD}, \VD\ and \RD\ problems with respect to the neighborhood diversity and/or modular-width parameters. }

\noindent {\bf Our results.} 
We give some positive and negative results with respect to 
some structural parameters of the input graph: modular-width, neighborhood diversity, and treewidth. The definitions of the parameters are given in Sections  \ref{hardness}, \ref{sec-mw}, and \ref{sec-tw}, respectively.
{It is worth mentioning that modular decomposition parameters, which comprise modular-width, neighborhood diversity, and tree-like parameters, such as treewidth and pathwidth, are two incomparable classes   that can be viewed as representing dense and sparse graphs, respectively.}
\begin{itemize}
\item We prove the W[1]-hardness of   \DVD\ with respect to neighborhood diversity even when all the radii are equal to $1$; namely, we show that  \VD\ is W[1]-hard with respect to neighborhood diversity.  
This negative result also applies to any generalization of neighborhood diversity and, in particular to modular-width and clique-width.
\item On the positive side,  we present FPT algorithms parameterized by:\\
{(i) Modular-width for \RD, with running time $O(\mw \; 2^\mw \;  n)$; \\
(ii) Modular-width plus the size $k$ of the solution  for \VD, with running time   $O(\mw \; k(k+1)^{\mw} \; n^2)$; \\
(iii) Modular-width plus the size $k$ of the solution  for \DVD, with running time   $O(\mw^2 \; k(k+1)^{2\mw} \; n^2)$; \\
(iv) Treewidth plus the maximum radius $\dmax=\max_{v \in V} \d_v$ for  \RD,  with  running time  \\
\hphantom{011} $O(\tw(2\dmax{+}1)^\tw   (\tw^2 {+} n) \; n^2 \log n)$;\\
(v)\, Treewidth plus $\tmax{=}\max_{v \in V} \t_v$ for  \VD,  with running time 
$O(\tw^2 2^{\tw}(\tmax{+}1)^{\tw} \; n).$ \\
\hphantom{011} This last result improves the above-described result achieved in \cite{IOU}}.
\end{itemize}

\begin{table}[tb!]
\begin{center}
\begin{tabular}{|l|l|l|l|}
\hline
 Parameters$\backslash$Problem \                & \DVD                  & \VD    & \RD  \\ \hline \hline

\nd            & \textbf{W{[}1{]}-hard} [Cor \ref{DVD-hard-nd}]        & \textbf{W{[}1{]}-hard} [Th \ref{w-nd}]      & \textbf{FPT}  [Th \ref{RD-mw}]                                     \\ \hline
\mw            & \textbf{W{[}1{]}-hard} [Cor \ref{DVD-hard-nd}]      & \textbf{W{[}1{]}-hard} [Th \ref{w-nd}]       & \textbf{FPT}  [Th \ref{RD-mw}]                                  \\ \hline
 \mw\ and $k$      & \textbf{FPT}  [Th \ref{DVD-mw-k}]                &\textbf{FPT}  [Th \ref{VD-mw-k}]                 & \textbf{FPT}  [Th \ref{RD-mw}]                              \\ \hline\hline
\tw            & W{[}1{]}-hard\cite{Betzler2012} \ & W{[}1{]}-hard\cite{Betzler2012}  & 
W{[}1{]}-hard\cite{KLP19}
\\ \hline
\tw\ and  $\dmax$&       open              & Not applicable ($\dmax=1$)      \          & \textbf{FPT}  [Th \ref{RD-tw-d} ]                                   \\ \hline
\tw\ and $\tmax$ &     open              & \textbf{FPT} \cite{IOU} [Th \ref{VD-tw-t}]                 &    
 Not applicable  ($\tmax=1$)   \\                                \hline
\end{tabular}
\end{center}
\caption{
Parameterized complexity results with respect to neighborhood diversity (\nd), modular-width (\mw)  and treewidth (\tw). 
}
\label{tableResults}
\end{table}

\section{Hardness} \label{hardness}
In this section, we prove that   \VD, and as a consequence its generalization \DVD,     is W[1]-hard on graphs of bounded neighborhood diversity. 

\smallskip

\noindent \textbf{Neighborhood diversity.}
Given a graph $G=(V,E)$,  two vertices $u,v\in V$  are said to have the same  
{\em same type} if $N_G(v) \setminus \{u\} = N_G(u) \setminus \{v\}.$\\ The {\em neighborhood diversity} of a graph $G$, introduced by Lampis  \cite{Lampis} and denoted by {\em \nd}$(G)$, is the smallest integer $\nd$ such that there exists a partition  $V_1, \ldots, V_\nd$, of the vertex set $V$, 
where all the vertices in  $V_i$ have the same type, 
for  $i\in [\nd]$\footnote{For a positive integer $a$, we use $[a]$ to denote the set of integers $[a] = \{1, 2, \ldots, a\}$.}.
The unique family 
$\{V_1, \ldots, V_\nd\}$ is called  the {\em type partition} of $G$.
%
\begin{theorem}\label{w-nd}  
  \VD\   is W[1]-hard with respect to 
  neighborhood diversity. 
\end{theorem} 
\proof
We use a reduction from 
{\sc{Multi-Colored clique (MQ)}}:
{\em Given   a graph $G=(V,E)$ and  a  proper vertex-coloring 
${\bf c} : V \to [ \q]$ for $G$, 
does $G$ contain  a  clique  of size $\q$?}

It is worth noticing that a multi-colored clique of size $q$ has one vertex of each color. Hence,    a vertex $v$ can belong to a multi-colored clique only if $ N_{G}(v)\cup\{v\}$ contains at least one vertex from each color class.  Hence, in the following, we will assume that all the vertices that do not satisfy such a property are removed from $G$ since they are irrelevant to the problem.

Given an instance $\langle G, {\bf c},\q \rangle$ of  {\sc MQ}, we construct $\langle G'=(V', E'),\bt,\one, k \rangle$, an instance  of the decision version of \VD.
Our goal is to guarantee that any solution of  \VD\  in $G'$ of size $k$ encodes a multi-colored clique in $G$ of size $q$ and vice-versa.

For a color $c \in [q]$, we denote by $V_c$ the class of vertices in $G$ of color $c$ and for a pair of distinct colors  $c,d \in [q],$ we let $\Ecd$ 
represent the edges in $G$ connecting a vertex in   $V_c$ and one in $V_d$.
We use the fact that  \MQ\  remains W[1]-hard 
even if each color class has the same size,  and between every pair of color classes we have the same number of edges   \cite{CFKLMPPS15}.
We then denote by $\r+1$ the size of each color class  
 $V_c$ 
 and by $\s + 1$ the size of each set $\Ecd$ (notice that $\r,\s \geq 1$ since otherwise $G$ is a clique).
We  use the following notation 
\begin{equation}\label{VE} 
V_c =\{\vco,\vcu, \ldots, \vcr\},\  \qquad\   \Ecd=\{\ecdo, \ldots, \ecdm\}
\qquad  \mbox{ $c,d \in [q]$, $c\neq d$}
\end{equation}
and refer to $\vci$ and $ \ecdj$ as the $i$-th vertex in $V_c$ and  $j$-th edge in $\Ecd$, respectively.

Let $\langle G,{\bf c}, \q \rangle$ be an instance of \MQ. 
 We describe a reduction from  $\langle G, {\bf c}, \q \rangle$  to an instance $\langle G', \bt,\one, k \rangle$ of \VD\  such that $\nd(G')$ is $O(\q^2)$.
To this aim, we introduce some gadgets 
for the construction of  $G'$, inspired by those used in \cite{DKT16}.
The rationale behind the construction is the following. First, we create two sets of gadgets (Selection and Multiple), which encode in  $G'$ the selection of vertices and edges as part of a potential multi-colored clique in $G$. Then we create another set of gadgets (Incidence gadgets) that is used to check whether the selected sets of vertices and edges 
form a multi-colored clique. Our goal is to guarantee that any solution of \VD\ in $G'$ encodes a multi-colored clique in $G$ and vice-versa.

In the following we  call   {\em \bug}
an independent set of vertices sharing all neighbors.
Hence, a connection between two \bug s  implies that the vertices in the  \bug s induce a complete bipartite graph.
We will also use {\em cliques}; the connection between two cliques is a complete bipartite graph among the vertices in the cliques.
Fig. \ref{reduction-nd} shows the gadgets we are going to introduce and how they are connected.

\begin{figure}[t!]
	\centering            \includegraphics[width=14truecm,keepaspectratio]{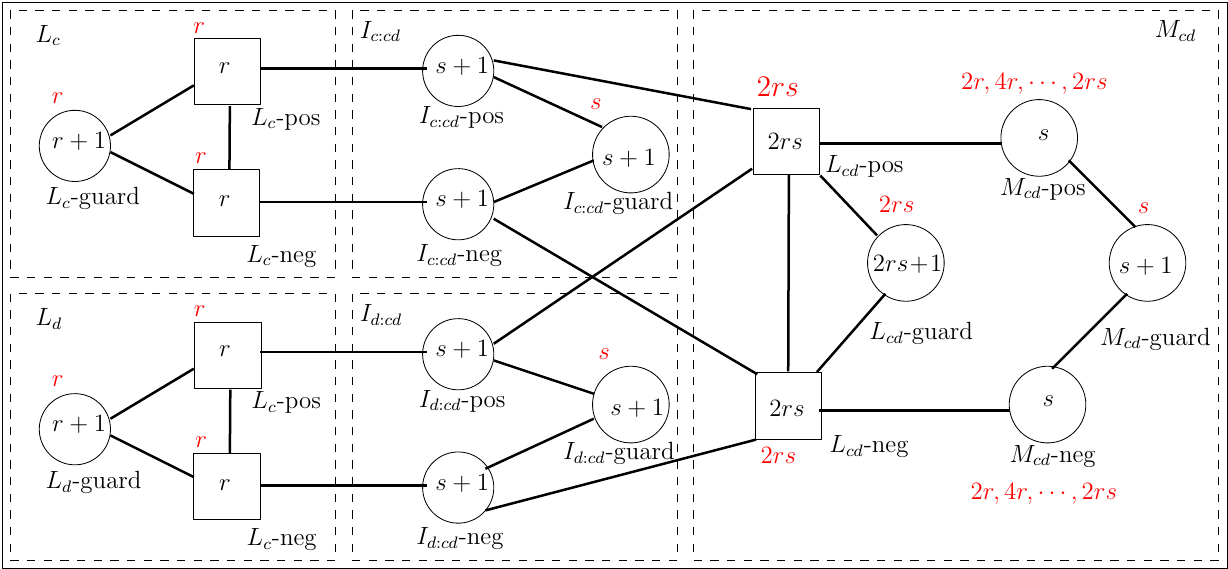} 
	\caption{\small An overview of the reduction. Each circle represents a \bug. Each square represents a clique. The number inside a \bug\ (resp. clique) is the number of vertices of
the \bug\ (resp. clique). The value $\t_v$ for a vertex $v$ is displayed in red. 
\label{reduction-nd}}
\end{figure}

\noindent
{\bf Selection Gadget.}
For each  $c \in [\q]$, 
the selection gadget $\Lc$ consists of two cliques, $\Lcn$ and  $\Lcp$ of $\r$ vertices each, and one \bug\  $\Lcg$
of $\r+1$ vertices. The cliques  $\Lcn$ and $\Lcp$ are connected, and the \bug\ $\Lcg$  is connected to both $\Lcn$ and $\Lcp$.
The selection gadget $\Lc$ is connected to the rest of the graph $G'$ using only vertices from $\Lcn\cup \Lcp$.
 We set now the value $\t_v=\r$ for each vertex $v \in \Lcp \cup \Lcn \cup \Lcg$.

\noindent
{\bf Multiple Gadget.}
For each $c, d \in [q]$ with $c\neq d$, 
we create a multiple gadget $\Mcd$ 
consisting of four \bug s, $\Lcdg$ of $2\r\s+1$ vertices, $\Mcdp$ and $\Mcdn$ of $\s$ vertices each, and 
$\Mcdg$ of $\s+1$ vertices, and two cliques $\Lcdp$ and $\Lcdn$ of $2\r \s$ vertices each. 
$\Mcdg$  is connected to  $\Mcdp$ and $\Mcdn$.
$\Mcdp$ is connected to $\Lcdp$, and $\Mcdn$ is connected to $\Lcdn$. The two cliques  $\Lcdp$ and $\Lcdn$ are connected.
Finally, the \bug\  $\Lcdg$ is connected to both $\Lcdp$ and $\Lcdn$.
The rest of graph $G'$ is connected only to the cliques
$\Lcdp$ and $\Lcdn$.
We set now the value $\t_v$ of each  $v$ in $\Mcd$ as: 
\begin{equation} \label{multiple}
\t_v=\begin{cases}
{2\r\s} &{\mbox{if $v\in \Lcdn \cup \Lcdp \cup \Lcdg$ }}\\
{2\r j} &{\mbox{if $v=x_j$ where $\Mcdp=\{x_1, \ldots, x_\s\}$}} \\
{2\r j} &{\mbox{if $v=y_j$ where $\Mcdn=\{y_1, \ldots, y_\s\}$ }}\\
{\s} &\mbox{if $v\in \Mcdg$}
\end{cases}
\end{equation}


\noindent
{\bf Incidence Gadget.} For each pair of distinct  $c,d \in [q]$, we construct two incidence gadgets: $\Iccd$ (connected with the gadgets $\Lc$ and $\Mcd$) and $\Idcd$ (connected with the gadgets $\Ld$ and $\Mcd$). In the following, we present the gadget $\Iccd$, which has the same structure as the gadget $\Idcd$.
The incidence gadget $\Iccd$  has three \bug s $\Iccdp$, $\Iccdn$  and  $\Iccdg$
of $\s + 1$ vertices each.
We connect $\Iccdg$ to $\Iccdp$ and $\Iccdn$. 
Furthermore, we connect $\Iccdp$ to $\Lcp$ and 
$\Lcdp$. Similarly, we connect  $\Iccdn$ to  
$\Lcn$ and $\Lcdn$.

Recalling that there are $\s + 1$ edges in the set $\Ecd$, and 
that there are $\s+1$ vertices in $\Iccdp$ and $\Iccdn$,
we create  one-to-one correspondences  between $\Ecd$ and $\Iccdp$
and between $\Ecd$ and $\Iccdn$.
Namely, for each $j=0,\ldots  s$, we associate the  $j$-th edge $\ecdj$  in $\Ecd$ (cfr. (\ref{VE})) to a   vertex  $u_j \in \Iccdp$ and to a vertex $w_j \in \Iccdn$ (with $u_j\neq u_{j'}$ and $w_j\neq w_{j'}$, for $j\neq j'$). Moreover, if the endpoint of $\ecdj$ of color $c$ is   the $i$th vertex  $\vci$ of $\Vc$ (cfr. (\ref{VE})) then  we set the value $\t_v$ of each vertex $v$ in $\Iccd$ as: 
 \begin{equation} \label{incident}
 \t_v=\begin{cases}
{ 2\r j+i} &{\mbox{if $v=u_j$ where $\Iccdp= \{u_0,\ldots, u_{\s}\}$}} \\
{2\r (\s -j) +\r -i} &{\mbox{if $v=w_j$ where $\Iccdn= \{w_0,\ldots, w_{\s}\}$}}\\
{s} &\mbox{if $v\in \Iccdg$}\\
\end{cases}
\end{equation}
It is worth observing that the vertices in $\Iccd$-pos (respectively, $\Iccd$-neg) have different demands. Indeed,  the numbers $2\r j+i$ (respectively, $2\r (\s -j)+\r-i$) are all  different, for $0\leq i \leq r$ and $0\leq j \leq s$.


The budget $k$ is set to  $k=\q \r +   \binom{\q}{2} (2\r + 3) \s$.
%
%
%
\begin{lemma} \label{clique-reduction}
$\langle G, {\bf c}, \q\rangle$ is a {\sc yes} instance of \MQ\ iff
$\langle G', \bt,\one, k\rangle$
is a {\sc yes} instance.
\end{lemma}
We complete the proof by showing  that 
$G'$ has neighborhood diversity $O(q^2)$.
Since each bag in $G'$ is a type set in the type partition of $G'$ and,  since for each $c\in [q]$, there are two cliques and one bag in $\Lc$ and, for each $c, d \in [q]$ with $c\neq d$ there are
four bags and two cliques in $\Mcd$, and three bags in both $\Iccd$ and $\Idcd$,
we have that the neighborhood diversity of $G'$ is $3 q + 12 \binom{q}{2}$.
\qed
\begin{corollary} \label{DVD-hard-nd}
  \DVD\   is W[1]-hard with respect to neighborhood diversity. 
\end{corollary}

\section{ FPT algorithms for graphs of bounded Modular-width } 
\label{sec-mw}
The notion of modular decomposition of graphs was introduced by Gallai in \cite{Gallai}, as a tool to define hierarchical decompositions of graphs.
{A {\em module} of a graph  $G = (V, E)$  is  a subgraph $G[M]$  induced by a  set $M \subseteq V$ 
such that all the vertices of $M$ share the same neighbors in $V\setminus  M$.}
The \textit{modular-width} parameter has been proposed in \cite{GLO}. 
\begin{definition}\label{mw}
{
Consider graphs obtainable by using (in any order or number) the following operations. 
\begin{itemize} 
\item (O1) \  The creation of an isolated vertex.
\item (O2) \ $G_1\oplus G_2$, called {\em disjoint union} of two graphs: $G_1\oplus G_2$ is the graph with vertex set $V (G_1 ) \cup V ( G_2 )$ and edge set $E( G_1 )\cup E( G_2 )$.
\item (O3) \ $G_1\otimes G_2$, called {\em complete join}: 
$G_1\otimes G_2$ is the graph with vertex set $V (G_1 )\cup V ( G_2 )$ and edges $E( G_1 )\cup E( G_2 )\cup \{ (u, w) \mid  u \in V ( G_1 ),\,  w\in  V(G_2) \}$.
\item (O4) 
$H(G_1,\ldots,G_\n)$, the substitution of the vertices $v_1,\ldots,v_\n$ of a graph $H$ by the graphs (modules) $G_1,\ldots,G_\n$: 
$H(G_1,\ldots,G_\n)$ is the graph with vertices $\bigcup_{1\leq \ell\leq \n} V(G_\ell)$ and edges
$\bigcup_{1\leq \ell \leq  \n} E(G_\ell)\cup \{(u, w) \mid u \in V (G_i), \ w \in
V (G_j ), \  (v_i, v_j) \in E(H) \}. $
\end{itemize}} 
\end{definition}

The \textit{modular-width} of a
graph $G$, denoted $\mw(G)$, is the least integer $\n$ such that $G$ can be obtained by using only the operations (O1)--(O4) (in any number and order) and where each operation (O4) has at most $\n$ modules.
A hierarchical decomposition of $G$ that is an expression using only the operations (O1)--(O4)  of width $\mw(G)$ can be constructed in linear time  \cite{TCH+}.\\
Notice that any module $\G_i$ of $\G =
\H(\G_1,\ldots,\G_\n)$  is such that all the vertices in $V(\G_i)$  have the same neighborhood in $V(G) \setminus V(G_i)$; that is,  for each vertex $u \in V(\G) \setminus V(G_i)$ either $V(G_i) \subseteq N_\G(u)$ or $V(G_i) \cap N_\G(u) = \emptyset$. Moreover,  operations (O2) and (O3) are special cases of (O4) for $H$ being $K_2$ or its complement. Hence, any graph $\G$ can be written as $\G =
\H(\G_1,\ldots,\G_\n)$ with $\n \leq \max\{2,\mw(\G)\}$.
Consider then the parse-tree of an expression describing   $\G$, according to the
operations (O1)--(O4). 
{The leaves of the parse-tree are the isolated vertex modules, created  by  (O1) and representing the vertices in $G$.}
Any internal vertex in the parse-tree is obtained through  (O2)-(O4):
Each such an operation corresponds to  a vertex  $\H(\G_1, \ldots,\G_\n)$ with $\n\geq 2$ children $\G_1, \ldots,\G_\n$.
%
{ 
\begin{observation}\label{fact1}
If $\G=\H(\G_1,\ldots,\G_\n)$  is a connected undirected graph, then $\dist_\G(u,v) \leq 2$, for each  $u,v \in V(\G_i)$ for any $i \in [\n]$.  
\end{observation}
}


\subsection{\RD\ with parameter \texorpdfstring{$\mw$}{mw}}\label{Algo} 
\begin{algorithm}[b!]
 \SetCommentSty{footnotesize}
  \SetKwInput{KwData}{Input}
 \SetKwInput{KwResult}{Output}
 \DontPrintSemicolon
  
\caption{  {\sc RD}$(\G=\H(\G_1, \ldots, \G_\n),  \bd, S_H$) \label{algdvd-1}}
\setcounter{AlgoLine}{0}
$S=\emptyset$\\
\lFor{\textbf{each} $i \in [\n]$ }{
              $\ell_i= \min_{j \in S_H\setminus\{i\}}\dist_H(i,j)$
              }
 \For{\textbf{each}  $i \in [\n] \setminus S_H$ }{
              \lIf{$\ell_i > \min_{v \in V(\G_i)}\d_v$}{
        \Return $(R{=}\false, S{=}\emptyset)$}
                    }
     \For{\textbf{each}  $i  \in S_H$ }{
           $A_i=\{v \in V(\G_i) \ | \ 1=\d_v < \ell_i\}$\\
              \lIf{$A_i = \emptyset $}{ $S=S \cup \{u_i\},$ where  $u_i$ is any vertex in $ V(\G_i)$}
              \Else{
              \lIf{$ \bigcap_{v \in A_i} \Ni{G_i}{v} = \emptyset$ 
              }{    \Return $(R{=}\false, S{=}\emptyset)$}
              \lElse{ $S=S \cup \{u_i\},$ where  $u_i$ is any vertex in $\bigcap_{v \in A_i} \Ni{G_i}{v} $}
                    }
                    }
        \Return $(R=\true, S)$
  \end{algorithm}

This section is devoted to proving the following result.
\begin{theorem} \label{RD-mw}
  {\RD\  can be solved in time $O(\mw \;2^\mw \;  n).$}
\end{theorem}
\begin{lemma} \label{one}
Let $\G=\H(\G_1,\ldots,\G_\n)$. 
There exists a solution $S$ for the instance
$\langle G, \bu, \bd\rangle$  of the \RD\ problem 
such that $|S \cap V(\G_i)| \leq 1$, for each $i \in [\n]$.
\end{lemma}

By exploiting Lemma \ref{one}, our algorithm proceeds by considering all the subsets $S_H \subseteq[\n]$, ordered by size, and checking whether it is possible to find a vertex $u_i \in V(\G_i)$ for each $i \in S_H$ such that $S=\{u_i\ | \ i \in S_H\}$ is a solution for  the instance
$\langle G, \bu, \bd\rangle$ of the \RD\ problem. Algorithm \ref{algdvd-1} implements this check. 

\begin{lemma}\label{lemma_3}
Given any set $S_H \subseteq [\n]$, 
 let $(R,S)$ be the pair returned by Algorithm  {\sc RD}$(\G=\H(\G_1, \ldots, \G_\n),  \bd, S_H$). If $R=\true$ then $S$ is a solution for the instance $\langle \G, \bu, \bd\rangle$ of the \RD\ problem, otherwise the problem has no solution with exactly one vertex selected from each $V(\G_i)$ with $i \in S_H$.
\end{lemma} 
Now we evaluate the running time of our algorithm. 
First of all,  we can obtain $ \min_{v \in V(\G_i)}\d_v$  for $i \in [p]$ in time $O(n)$. Then,  for at most each $S_H \subseteq[\n]$, we use algorithm  {\sc RD}$(\G,  \bd, S_H$) to verify if a solution $S$  with exactly one vertex selected from each $V(\G_i)$ with $i \in S_H$ exists. Considering that Algorithm {\sc RD}$(\G,  \bd, S_H$) requires time $O(\n \; n)$ and that the number of modules of $G$ is  $\n \leq \mw$, overall we have time complexity $O(\mw \;2^\mw \;  n).$ \qed

\subsection{\VD\  with parameters \texorpdfstring{$\mw$}{\mw} and the solution size \texorpdfstring{$k$}{k}} %
\begin{algorithm}[b!]
\SetCommentSty{footnotesize}
\SetKwInput{KwData}{Input}
\SetKwInput{KwResult}{Output}
\DontPrintSemicolon
\caption{{\sc VD-MW}$(\hG=\hH(\hG_1, \ldots, \hG_\hn), \hbt,   \hb$) \label{alg-vd-mw}}

\setcounter{AlgoLine}{0}
\If(\tcp*[h]{ $\hH$ is a single vertex graph}) {$\hG=\hH(\hG_1)=(\{v\},\emptyset)$ }{
          \lIf{$\hb=0 \land  \t_v\geq 1$ }{\Return $(R=\false,\hS=\emptyset)$}
         \lIf{$\hb=0 \land \t_v=0$}{\Return $(R=\true,\hS=\emptyset)$}
       \lIf{$\hb=1$}{\Return $(R=\true,\hS=\{v\})$}
     }
    \Else{
      \For{\textbf{each}  $(s_1, \ldots, s_\hn) \mid \sum_{i=1}^\hn s_i=\hb$  and $0 \leq s_i \leq \min \{\hb,|V(\hG_i)|\}$}{
             \For{$i=1, \ldots, \hn$}{
                 \lFor{\textbf{each}  $v \in V(\hG_i)$}{
                      $\htt'_{v} = \max \{0, \ \htt_{v} - \sum_{j \mid (i,j)\in E(\hH)}s_j \}$
                       } 
                  $(R_i,S_i)=${\sc VD-MW}$(\hG_i, \hbt',  s_i)$
                   }
        \lIf{$\bigwedge_{i=1}^\hn R_i =\true$}{\Return $(R =\true, \; \hS = \bigcup_{i=1}^\hn S_i)$} 
    }       
    \Return $(R=\false,\hS=\emptyset)$
   }  
  \end{algorithm}

This section is devoted to proving the following result.
\begin{theorem} \label{VD-mw-k} 
   \VD\  can be solved in time {$O(\mw \; k(k+1)^{\mw} \;n^2)$.}
\end{theorem}
Consider the parse-tree of an expression describing the input graph $G$, according to the
operations (O1)-(O4) in Definition \ref{mw}. 
We design a recursive algorithm that computes a {\em  Vector Dominating} set for the instance $\langle G, \bt,  \bu\rangle$ based on the parse-tree of $G$.

Except for the leaves of the parse-tree (representing (O1)) and thus graphs consisting of exactly one vertex, i.e.,   $\hH(\hG_1)=(\{v\},\emptyset)$), for all the other vertices of the parse-tree we just need
to focus on the operation (O4), that is $\hG = \hH(\hG_1, \ldots, \hG_\hn)$ such that $\hn \leq \max\{2,\mw(G)\}.$

For the instance $\langle G, \bt,\bu\rangle$ of the \VD\ problem, our algorithm checks if there exists a solution for the decision version of the problem, with instance $\langle G, \bt,\bu,  b\rangle$, that asks for a  {\em  Vector Dominating} set of size $b$ of $G$ with respect to the demand vector $\bt$.
The minimum positive integer $b$ for which the instance $\langle G, \bt,\bu, b\rangle$ has a solution is the size $k$ of the solution of the instance $\langle G, \bt, \bu\rangle$ of the \VD\ problem.  
The algorithm uses a recursive approach along the parse-tree of $G$ and for each vertex $\hG=\hH(\hG_1, \ldots, \hG_\hn)$ of the parse-tree and the relative instance $\langle \hG, \hbt, \bu, \hb\rangle$ with $\hb \leq |V(\hG)|$, constructs an equivalent instance of the problem on each $\hG_i$ obtained by partitioning the budget $\hb$ among the $\hn$ modules $\hG_1, \ldots, \hG_\hn$ and appropriately reducing the values in the demand vector. 
The solution set $\hS$ for $\langle \hG, \hbt, \bu, \hb\rangle$ is reconstructed by using the solutions recursively obtained for each $\hG_i$ (cf. Algorithm \ref{alg-vd-mw}).
\subsection{\DVD\  with parameters \texorpdfstring{$\mw$}{\mw} and the solution size \texorpdfstring{$k$}{k}}
 \begin{algorithm}[b!]
 \SetCommentSty{footnotesize}
  \SetKwInput{KwData}{Input}
 \SetKwInput{KwResult}{Output}
 \DontPrintSemicolon
\caption{{\sc DVD-MW}$(G=H(G_1, \ldots, G_\n), \bt, \bd,  b$) \label{alg-dvd-mw}}
\setcounter{AlgoLine}{0}
 \For{\textbf{each}  $(s_1, \ldots, s_\n) \mid \sum_{i=1}^\hn s_i=b$  and $0 \leq s_i \leq \min \{b,|V(G_i)|\}$}{
   \For{$i=1, \ldots, \n$}{
        \lIf{  $\exists \ v \in V(G_i)$   :  
              ($d_v \geq 2$)  and ($t_{v} - s_i - \sum_{j \mid \; j\neq i \; \land \; \dist_\H(i,j)\leq \d_v}s_j >0)$ }{\Return $(R=\false,S=\emptyset)$}
        \For{\textbf{each}  $v \in V(G_i)$}{ $t'_{v} = \begin{cases}
             0  & \mbox{if } d_v \geq 2\\
            \max \{0, \ t_{v} - \sum_{j \mid (i,j) \in E_H}s_j \} & \mbox{if } d_v=1
        \end{cases}$ }
                   $(R_i,S_i)=${\sc VD-MW}$(G_i, \bt', s_i$) 
                   }
        \lIf{$\bigwedge_{i=1}^\n R_i =\true$}{\Return $(R =\true, \; S = \bigcup_{i=1}^\n S_i)$}  
          }
    \Return $(R=\false,S=\emptyset)$
  \end{algorithm} 

In this section we present an algorithm to solve the \DVD\ problem by using the algorithm {\sc VD-MW} given in the previous section.
We prove the following result. 
\begin{theorem} \label{DVD-mw-k} 
   \DVD\  can be solved in time 
   {$O(\mw^2 \;k (k+1)^{2\mw} \;n^2)$.}
\end{theorem}

Let $\G=\H(\G_1, \ldots, \G_\n)$ be the input graph and let  $\langle G,  \bt,\bd, k\rangle$ be an instance of the decision version of the \DVD\ problem, asking for a  {\em  Distance Vector Dominating} set of size $k$ of $G$ with respect to the demand vector $\bt$ and the radius vector $\bd$.
Our algorithm {\sc DVD-MW}, which checks for a solution for the decision version of the problem with instance $\langle G,\bt, \bd,   b \rangle$,  is based on the following easy considerations, for any vertex $v \in V(\G_i)$ and $i \in [\n]$:\\
--  If $d_v \geq 2$ then any vertex in $V(\G_i)$ that is selected in the solution dominates $v$ (recall Observation \ref{fact1}) together with any vertex in $V(\G_j)$ that is selected in the solution, for $j$ such that $j\neq i$ and $\dist_\H(i,j)\leq \d_v$.\\
-- If $d_v=1$ then any vertex in $V(\G_j)$ with $(i,j) \in E(H)$, that is selected in the solution, dominates $v$.

Consider a partition $(s_1, \ldots, s_\n)$ of $b$  and select $s_i$ vertices in  $G_i$, for each $i \in [\n]$. If there exists $v \in V(\G_i)$ with $d_v \geq 2$ and demand $t_v > s_i + \sum_{j \mid \; j\neq i \; \land \; \dist_\H(i,j)\leq \d_v}s_j$ then the partition has to be discarded (since there are not enough selected vertices to dominate $v$). Otherwise, all the vertices with $d_v \geq 2$ are dominated by any choice of vertices   satisfying the  partition and we  only have to worry about each vertex $v$ with $d_v =1$. In particular, the selection in each $V(\G_i)$ has to be accurate in order to have a vector dominating set for   $\langle \G_i, \bt',\one, s_i\rangle$, where $t'_v$ is defined in Algorithm \ref{alg-dvd-mw}. 

\section{ FPT algorithms for graphs of bounded treewidth}\label{sec-tw}
\label{secLBTW} 
\begin{definition}\label{treedec} 
A tree decomposition of a graph $G=(V,E)$ is  a pair $(T,\{W_i\}_{i \in V(T)})$, where $T$ is a tree and   each  $i$ in $T$ is assigned a 
$W_i \subseteq V$ such that: 
\begin{itemize} 
\item[1.] 
  $ \bigcup_{i \in V(T)} W_i=V$.
\item[2.]  For each  $e=(v,u)\in E,$ there exists  $i$ in $T$ s.t. $W_i$ contains both $v$ and $u$.
\item[3.]  For each $v \in V,$  the tree induced by 
 $T_v=\{ i \in V(T) \mid v \in W_i\}$
is connected. 
\end{itemize}
\end{definition} 
  The width of a tree decomposition  $(T,\{W_i\}_{i \in V(T)})$ of a graph $G$, is defined as $\max_{i \in V(T)}|W_i|-1$.
The treewidth of  $G$, denoted by $\tw(G)$, is the minimum width over all tree decompositions of $G$.
{Deciding whether a graph has tree decomposition of treewidth at most $k$ is NP-complete \cite{ACP87} and 
proved fpt in \cite{BK91}.}

\begin{definition}
\label{nice}
\cite{Treewidth_book} \ A 
tree decomposition $(T,\{W_i\}_{i \in V(T)})$ is called nice if  it satisfies conditions 1. and 2.:
\begin{itemize} 
\item[1.] $W_r=\emptyset,$ for $r$ the root of $T$ and $W_i=\emptyset,$ for every leaf $i$ of $T$.
\item[2.] Every non-leaf vertex of $T$ is of one of the following three types: \\
 \textbf{Introduce:} a vertex $i$ with  one child $j$ s.t. $W_i=W_{j}{\cup} \{v\}$ for a vertex $v\notin W_{j}$.\\
	 \textbf{Forget:} a vertex $i$ with  one child $j$ s.t.  $W_{j}=W_i\cup \{v\}$ for a vertex $v\notin W_{i}$.\\
	 \textbf{Join:} a vertex $i$ with two children $i_1,i_2$ s.t. $W_{i}=W_{i_1}=W_{i_2}.$
\end{itemize}
\end{definition} 
Consider a graph  $G=(V,E)${. Given a tree decomposition of $G$ of width $\tw$,} one can compute in polynomial time a nice tree decomposition $(T,\{W_i\}_{i \in V(T)})$ of $G$ of treewidth at most $\tw$ having $O(\tw|V (G)|)$ vertices \cite{Treewidth_book}. 
Let  $T$ be rooted in  $r$. For any  $i$ in $T,$   denote by $T(i)$ the subtree of $T$ rooted at  $i$,  by  $W(i)=\bigcup_{j \in T(i)} W_j$  the union of  the bags in $T(i)$,  and  by  
$s_i = \lvert W_i\lvert $ the size of $W_i$.
\def\C{{\cal C}}

\def\vd{L}
\def\vD{{\cal \vd}}
\def\K{{\cal K}}

\smallskip
\noindent \textbf{A FPT algorithm parameterized by \texorpdfstring{$\tw$}{tw} plus 
$\tau$ for  \VD.}
We give a dynamic programming algorithm which, exploiting a nice tree decomposition, 
recursively solves 
the \VDl\ (\VD) problem.
Fix  $i \in V(T),$ to  recursively reconstruct the solution, we  calculate  optimal solutions under different  hypotheses based on the following considerations:
 For each vertex $v \in W_i$ we have two cases: $v \in S$, $v\notin S$. We are going to consider all the $2^{s_i}$ combinations of such states with respect to some solution $S$ of the problem. We denote each combination with a binary vector $\vd_i$ of size $s_i$ indexed by the elements of $W_i$, where for each $v \in W_i$, $\vd_i(v)=1$, if $v \in S$ and $\vd_i(v)=0$, otherwise.  The configuration $\vd_i = \emptyset$ denotes the vector of length $0$ corresponding to an empty bag. We denote by $\vD_i$ the family of all the $2^{s_i}$ possible state vectors of the $s_i$ vertices in $W_i$.\\
We 
consider all the possible contributions to the \VD\ problem, of vertices in $V \setminus
W(i)$; that is, for each $v \in  W_i$, we consider all the possible demands among
$\t_v, \t_v {-} 1,\ldots,0$. 
 As a consequence, we will have up to $(\tmax{+}1)^{s_i}$ demand combinations, where $ \tmax {=} \max_{v \in V} \t_v$. We denote each possible demand combination with a vector $K_i$, indexed by the $s_i$ elements in $W_i$. The configuration $K_i = \emptyset$ denotes the demand vector of length $0$ corresponding to an empty bag. Moreover,  $\K_i$ represents the family of all the possible demand combinations of vertices in $W_i$.
 
The following definition introduces the values that will be computed by the algorithm in order to keep track of all the above cases.
 
 \begin{definition}\label{eq_Bi} 
 For each vertex $i\in V(T),$ each $\vd_i\in \vD_i$ and each $K_i \in \K_i$, we define
 $B_i(\vd_i,K_i)$ as the minimum number of vertices to be selected in $G[W(i)]$ in order to dominate all the remaining vertices in $G[W(i)]$, where the states and the demands of vertices in $W_i$ are given by $\vd_i$ and $K_i$.
 \end{definition}
 By noticing that the root $r$ of a nice tree decomposition has  $W_r = \emptyset$, we have that the solution of the \VD\ problem $\langle G, \bt,  \one \rangle$  can be obtained by computing  $B_r(\emptyset,\emptyset).$ 
  \begin{lemma} \label{claimTW2}
 For each  $i\in T$, the computation of $B_i(\vd_i, K_i)$,  for each    $\vd_i\in \vD_i$  and  $K_i\in \K_i$, comprises $O(2^{s_i}(\tmax+1)^{s_i})$ values,  each of which can be computed recursively in time $O(s_i)$.
 \end{lemma}
\begin{theorem}    \label{VD-tw-t} 
{If a tree decomposition of $G$ with width $\tw$ is given then}  \VD\ is solvable in  time $O(\tw^2 2^{\tw}(\tmax+1)^{\tw} \; n)$.
\end{theorem}
\begin{proof}{}
The decomposition tree has at most $O(\tw \; n)$ vertices \cite{Treewidth_book}. Hence, the desired value  $B_r(\emptyset,\emptyset)$, which
corresponds to the solution of the \VD\ instance $\langle G, \bt,  \one \rangle$, can be computed in time
$O(\tw^2 2^{\tw}(\tmax+1)^{\tw} \; n).$ The optimal set $S$ can be computed in the same time by
standard backtracking technique.
\qed
\end{proof}
\noindent \textbf{A FPT algorithm parameterized by $\tw$ plus $\delta$ for \RD.}
Exploiting a nice tree decomposition of the input graph $G$ and a strategy similar to the one adopted in \cite{Borradaile2016} we obtain the following result.
\begin{theorem} \label{RD-tw-d}
{If a tree decomposition of $G$ with width $\tw$ is given then} \RD\  is solvable in  time $O(\tw(2\dmax+1)^\tw  (n + \tw^2)\; n^2 \log n)$.
\end{theorem} 

\section{Concluding remarks} 
{We introduced the Distance Vector Domination problem which generalizes both
distance and multiple domination, at individual (i.e., vertex) level.
The problem is motivated by the development of strategies to mitigate the spread of fake information. Indeed the set identified by the problem can be used to detect a set of individuals who, disseminating debunking information, can prevent the spreading, of misinformation.
We analyzed the parameterized complexities of the problem according to several standard and structural parameters. } {It eluded us the design of an FPT algorithm for the DVD problem parameterized by the combination of treewidth and one of the other problem parameters, such as the size $k$ of the solution, the largest demand $\tau$ or the largest radius $\delta$, which we leave as an open problem. Additionally, it would be interesting to investigate the complexity of the RD problem with respect to the clique-width parameter.}



\begin{thebibliography}{8}
\bibitem{barabasi}
R. Albert, H. Jeong, A.-L. Barab\'asi.
Error and attack tolerance of complex networks.
Nature 404, 378--382, (2000).

\bibitem{ACP87}
 S. Arnborg, D.G. Corneil, A. Proskurowski. 
 Complexity of finding embeddings in a k-tree.
 SIAM J. Alg. Disc. Meth. 8,  277--284, (1987).

\bibitem{BJP}
S. Banerjee, M. Jenamani, D.K. Pratihar.  A survey on influence maximization in a social network. Knowl Inf Syst 62, 3417--3455, 10.1007/s10115-020-01461-4,  (2020).

\bibitem{Ben-Zwi}
O. Ben-Zwi, D. Hermelin, D. Lokshtanov, I. Newman.
Treewidth governs the complexity of target set selection.
Discrete Optimization 8(1), 87--96,  (2011).



\bibitem{sensor1}
J.-C. Bermond, L. Gargano,  A.A. Rescigno. 
Gathering with minimum delay in tree sensor networks.
Proc. of  SIROCCO'08,  LNCS 5058, 262--276, (2008).

\bibitem{BK91}
H.L. Bodlaender and T. Kloks. 
Better algorithms for the pathwidth and treewidth of graphs. 
Proc. of  ICALP'91,  LNCS 510,  544--555, (1991).


\bibitem{Betzler2012}
N. Betzler, R. Bredereck, R. Niedermeier and J. Uhlmann.
On Bounded-Degree Vertex Deletion parameterized by treewidth.
Discrete Applied Mathematics 160(1-2),  53--60, 
(2012).

\bibitem{Borradaile2016}
G. Borradaile and H. Le.
Optimal Dynamic Program for r-Domination Problems over Tree Decompositions. In Proc. of 11th International Symposium on Parameterized and Exact Computation, {IPEC}, LIPIcs  63, 8:1--8:23, (2016).



\bibitem{sensor2}
 C.-Y. Chong,  S.P. Kumar.
Sensor networks: Evolution, opportunities, and challenges.
 Proceedings of the IEEE, 91 (8),  1247--1256, (2003). 

\bibitem{CMV}
F. Cicalese, M. Milani\"c, U. Vaccaro.
On the approximability and exact algorithms for vector domination and related problems in graphs. 
Discrete Applied Math. 161, 750--767, (2013).


\bibitem{CCG+}
F. Cicalese, G. Cordasco, L. Gargano, M. Milanic and U. Vaccaro. Latency-bounded target set selection in social networks. Theoretical Computer Science 535,  1--15,  (2014).


  \bibitem{CGMRVa} G. Cordasco, L. Gargano, M. Mecchia, A. A. Rescigno, U. Vaccaro. A fast and effective heuristic for discovering small target sets in social networks. In Proc. of the 9th Inter. Conf. on Comb. Opt. and Appl. (COCOA), (2015). 
 \bibitem{asonam}
 G. Cordasco, L. Gargano, A. A. Rescigno.
 Influence propagation over large scale social networks.
 In Proc. of IEEE/ACM Inter. Conf. on Advances in Social Networks Analysis and Mining (ASONAM'15), 1531--1538,  (2015).

 \bibitem{SNAM}
G. Cordasco,  L. Gargano and A.A. Rescigno.
\newblock On finding small sets that influence large networks.
\newblock Social Network Analysis and Mining 6(11),  (2016).

\bibitem{CGMRV} G. Cordasco, L. Gargano, M. Mecchia, A.A. Rescigno, U. Vaccaro. Discovering Small Target Sets in Social Networks: A Fast and Effective Algorithm. Algorithmica 80(6),  1804--1833, (2018).



\bibitem{CGRV}	G. Cordasco, L. Gargano, A. A. Rescigno, U. Vaccaro. 
Evangelism in social networks: Algorithms and complexity.  Networks  71(4): 346--357, (2018).




\bibitem{active} 	G. Cordasco, L. Gargano, A. A. Rescigno. Active influence spreading in social networks. Theoretical Computer Science 764,  15--29, (2019).



\bibitem{CGL+} 
G. Cordasco, L. Gargano, M. Lafond, L. Narayanan, A. A. Rescigno, U. Vaccaro, K. Wu.
Whom to befriend to influence people.
Theoretical Computer Science 810,  26--42, (2020).


 \bibitem{itp} G. Cordasco, L. Gargano, A.A. Rescigno. 
Parameterized complexity for iterated type partitions and modular-width. In Discrete Applied Mathematics, 350,  (2024).

 \bibitem{TCH+}
D. Corneil, M. Habib, C. Paul and M. Tedder.
A recursive linear time modular decomposition algorithm via LexBFS.
arXiv:0710.3901[cs.DM], (2024).


\bibitem{C}
B. Courcelle, The monadic second-order logic of graphs recognizable sets of finite graphs, Inform. and Comput., 85 (1), 12–75, (1990).



\bibitem{CO00}
B. Courcelle and S. Olariu. Upper bounds to the clique width of graphs. Discrete
Applied Mathematics, 101(1), 77--114, (2000).

 \bibitem{CFKLMPPS15}
M. Cygan, F.V. Fomin, L. Kowalik, D. Lokshtanov, D. Marx, M. Pilipczuk, M. Pilipczuk and S. Saurabh.
Parameterized Algorithms. Springer,  (2015).


\bibitem{sensor3}
K. Dasgupta, M. Kukreja,   K. Kalpakis.
Topology-aware placement and role assignment for energy-efficient information gathering in sensor networks. In Proc. of IEEE Symposium on Computers and Communications, 341--348,  (2003).




\bibitem{DF}
R.G. Downey, M.R. Fellows. Parameterized Complexity. Springer, Heidelberg, (1999).




\bibitem{DKT16}
P. Dvor\'ak, D. Knop, and T. Toufar. 
Target Set Selection in Dense Graph Classes.  
In Proc. of 29th International Symposium on Algorithms and Computation  (ISAAC'18), 10.4230/LIPIcs.ISAAC.2018.18, (2018).

\bibitem{Flink85}
J. F. Fink and M. S. Jacobson.
n-Domination in graphs. Graph theory with applications to algorithms and computer science. John Wiley \& Sons,  283--300 (1985).

\bibitem{GLO}
J. Gajarský, M. Lampis, S. Ordyniak.
Parameterized Algorithms for Modular-Width. In Proc. of 
8th International Symposium on Parameterized and Exact Computation (IPEC 2013), LNCS 8246, 163--176, (2013).


\bibitem{Gallai} T.  Gallai.  
Transitiv orientierbare Graphen.
In Acta Mathematica Academiae Scientiarum Hungarica  18, 26--66, (1967).

\bibitem{ganian12}
R. Ganian.
Using Neighborhood Diversity to Solve Hard Problems.
In arXiv 2012, arXiv:1201.3091, (2012).



\bibitem{granovetter}
M. Granovetter. 
Threshold models of collective behaviors. 
The American Journal of Sociology, 83(6), 1420--1443, (1978).

\bibitem{GH} 
W. Goddard, M.A. Henning, Restricted domination parameters in graphs. Journal of Combinatorial Optimization 13 353--363,  (2007).



\bibitem{HHS1} 
T.W. Haynes, S. Hedetniemi, P. Slater.
Fundamentals of Domination in Graphs, Marcel Dekker, (1998).

\bibitem{HHS2} 
T.W. Haynes, S. Hedetniemi, P. Slater (Eds.). Domination in Graphs: Advanced Topics, Marcel Dekker, (1998).


\bibitem{HPV} 
J. Harant, A. Prochnewski, M. Voigt. On dominating sets and independent sets of graphs, Combinatorics. Probability and Computing 8,  547--553, (1999).

\bibitem{Henning20}
 M.A. Henning. Distance Domination in Graphs. In Topics in Domination in Graphs. Developments in Mathematics, 64, 10.1007/978-3-030-51117-3\_7, (2020).

\bibitem{SETH}
R. Impagliazzo and R. Paturi.
On the Complexity of k-SAT.
Journal of Computer and System Sciences,
 62(2), 367--375, 10.1006/jcss.2000.1727,
(2001).

\bibitem{IOU}
T. Ishii, H. Ono, Y. Uno.
(Total) Vector domination for graphs with bounded branchwidth.
Discrete Applied Mathematics,
207, 80--89, (2016).

\bibitem{KLP19}
I. Katsikarelis, M. Lampis, V. Th. Paschos.
Structural parameters, tight bounds, and approximation for (k,r)-center.
Discr. App. Math.,
264, 90--117, (2019).




\bibitem{Kempe}
 D. Kempe, J. Kleinberg, E. Tardos.  
Maximizing the spread of influence through a social network. 
 In Proc. of KDD'03,  137--146, (2003).

\bibitem{Kimuraetal}
M. Kimura,  K. Saito, H.  Motoda.
Blocking links to minimize contamination
spread in a social network. ACM Trans. on Knowledge Discovery from Data
3(2),  (2009).

\bibitem{Treewidth_book}
T. Kloks.
Treewidth Computations and Approximations.
LNCS 842, Springer-Verlag Berlin Heidelberg, ISSN 0302-9743,  10.1007/BFb0045375, (1994).

\bibitem{LL}
M. Lafond, W. Luo. 
Parameterized Complexity of Domination Problems Using Restricted Modular Partitions. 
In MFCS 2023,  61:1--61:14, (2023).


\bibitem{MP}
R. Lamblet Mafort, F. Protti.
Vector Domination in split-indifference graphs.
Information Processing Letters, 155,
ISSN 0020-0190, (2020).

\bibitem{Lampis}
M. Lampis.
Algorithmic meta-theorems for restrictions of treewidth.
Algorithmica 64, 19--37, (2012). 


\bibitem{LWS} 
P. Li, A. Wang, J. Shang. A simple optimal algorithm for k-tuple dominating problem in interval graphs.  Journal of Combinatorial Optimization 45(14), (2023). 

\bibitem{newman} M. E. J. Newman, S. Forrest, J. Balthrop.
Email networks and the spread of computer viruses.
Physical Review E 66, (2002).

\bibitem{RSS} 
V. Raman, S. Saurabh, S. Srihari, Parameterized algorithms for generalized domination. In Proc of COCOA 2008, LNCS 5165,   116--126, (2008).

\bibitem{R} 
M. Romanek. Parameterized algorithms for modular-width. Bachelor's Thesis, (2016).

\bibitem{Slater}
P. J. Slater, R-domination in graphs. J. Assoc. Comp. Mach. 23(3), 446--450, (1976).









\bibitem{Wolsey}
L.A. Wolsey. An analysis of the greedy algorithm for the submodular set covering problem.
Combinatorica 2, 385–393 (1982)
\end{thebibliography}
\end{document}